\documentclass[aps,pra,twocolumn,reprint,superscriptaddress,nofootinbib]{revtex4-2}

\usepackage[svgnames,table,xcdraw]{xcolor}
\usepackage{amsmath}
\usepackage{microtype}
\usepackage{amssymb}
\usepackage{braket}
\usepackage{graphicx}
\usepackage{booktabs}
\usepackage{bm}
\usepackage{dsfont}
\usepackage{hyperref}
\usepackage[caption=false]{subfig}
\usepackage{amsthm}
\newtheorem*{claim}{Claim}
\usepackage{xfrac}
\usepackage{enumitem}
\usepackage{orcidlink}
\usepackage{appendix}
\usepackage[ruled,vlined]{algorithm2e}

\hypersetup{colorlinks,linkcolor={DarkBlue},citecolor={DarkBlue},urlcolor={DarkBlue}}

\let\originalleft\left
                   \let\originalright\right
\renewcommand{\left}{\mathopen{}\mathclose\bgroup\originalleft}
\renewcommand{\right}{\aftergroup\egroup\originalright}

\usepackage{mathtools}

\DontPrintSemicolon
\SetKwInput{KwInput}{Input}
\SetKwInput{KwOutput}{Output}
\SetKwFor{ForEach}{foreach}{do}{}
\SetKwFor{ForLoop}{for}{do}{}

\usepackage{nag}
\usepackage{graphicx}
\usepackage{amsthm}
\usepackage{tabulary}
\usepackage{qcircuit}
\usepackage{mathtools}
\usepackage{bm}
\usepackage{braket}
\usepackage{datetime}
\usepackage{stackrel}
\usepackage{dsfont}
\usepackage{qcircuit}
\usepackage[english]{babel}
\usepackage{units}

\usepackage{tikz}
\usetikzlibrary{backgrounds,decorations.pathreplacing}

\usepackage{chngcntr}
\counterwithout{equation}{section}

\usepackage[normalem]{ulem}
\usepackage{slashed}
\usepackage{soul}

\usepackage{xcolor}
\usepackage{amssymb}
\usepackage{amsmath}
\usepackage{amsthm}
\usepackage{braket}
\usepackage{mathdots}

\usepackage{makeidx}
\makeindex

\usepackage{soul}
\usepackage{xargs}
\newcommandx{\cmnote}[2][1=]{\linespread{1.0}\todo[linecolor=red,backgroundcolor=red!25,bordercolor=red,#1]{#2}}

\let\underline\ul

\allowdisplaybreaks[4]

 \index{} \index{} \index{} \index{} \index{} \index{} \index{} \index{}%%
\makeatletter

\newcommand{\ringplus}{\mathbin{\text{\@ringplus}}}

\newcommand{\@ringplus}{%
  \ooalign{\hidewidth\raise1.3ex\hbox{\tiny$\circ$}\hidewidth\cr$\m@th+$\cr}%
}

\newcommand{\ringminus}{\mathbin{\text{\@ringminus}}}

\newcommand{\@ringminus}{%
  \ooalign{\hidewidth\raise0.9ex\hbox{\tiny$\circ$}\hidewidth\cr$\m@th-$\cr}%
}
\makeatother
 \index{} \index{} \index{} \index{} \index{} \index{} \index{} \index{}%%

\DeclareFontFamily{U}{wncy}{}
\DeclareFontShape{U}{wncy}{m}{n}{<->wncyr10}{}
\DeclareSymbolFont{mcy}{U}{wncy}{m}{n}
\DeclareMathSymbol{\Sh}{\mathord}{mcy}{"58}

 \index{} \index{} \index{} \index{} \index{} \index{} \index{} \index{}%%
 \index{} \index{} \index{} \index{} \index{} \index{} \index{} \index{}%%

 \index{} \index{} \index{} \index{} \index{} \index{} \index{} \index{}%%
 \index{} \index{} \index{} \index{} \index{} \index{} \index{} \index{}%%
\usepackage{graphicx}
\usepackage{graphicx}
\usepackage{multirow}
\usepackage{hhline}
\xyoption{color}
\xyoption{line}

\newcommandx*\bsbal[3][1=black, 3=->]{\ar @[#1]@{#3} [#2,0] \qw}

\newcommandx*\varbs[5][1=black, 3=\theta,4=0.5,5=->]{\ar @[#1]@{#5}^(#4){#3} [#2,0] \qw}

\newcommandx*\lblline[3][3=0.5]{\ar @{-}^(#3){#1} [#2,0]}
\newcommandx*\ctrlg[3][3=0.5]{ \raisebox{-3pt}{$\bullet$}  \ar @{-}^(#3){#1} [#2,0] \qw }
\newcommandx*\ctrlog[2]{\controlo \ar @{-}^{#1} [#2,0] \qw}
\newcommandx*\ctrlodash[1]{\controlo \ar @{-} [#1,0] \ar @[black]@{.} [0,-1]}

\begin{document}

  \title{%
    \texorpdfstring
    {Decoder Dependence in Surface-Code Threshold Estimation \\under Digitized Hybrid Continuous-Variable and Discrete Noise}
    {Decoder Dependence in Surface-Code Threshold Estimation under Digitized Hybrid Continuous-Variable and Discrete Noise}
  }

  \def \affGatech {College of Computing, Georgia Institute of Technology, Atlanta, GA 30332 USA}
  \def \vgg {Volkswagen AG, Berliner Ring 2, Wolfsburg 38440, Germany}
  \def \rwth {Department of Physics, RWTH Aachen, Germany}
  \def \fujf {Department of Computational and Applied Mechanics, \\Federal University of Juiz de Fora, Juiz de Fora, 36036-900, Brazil}
  \def \tubaf {Institute of Computer Science, Faculty of Mathematics and Computer Science, TU Bergakademie Freiberg, Bernhard-von-Cotta-Straße 2, D-09599 Freiberg, Germany}

  \author{Dennis Delali Kwesi Wayo \orcidlink{0000-0001-9980-6247}}
  \affiliation{\affGatech}
  \affiliation{\tubaf}
%    \email{dwayo3@gatech.edu}

  \author{Chinonso Onah, \orcidlink{0000-0002-6296-533X}}
  \affiliation{\vgg}
  \affiliation{\rwth}
%\email{chinonso.calistus.onah@volkswagen.de}
%    \email{calistusnonso@gmail.com}

  \author{Leonardo Goliatt \,\orcidlink{0000-0002-2844-9470}}
  \affiliation{\fujf}
%    \email{goliatt@gmail.com}

  \author{Sven Groppe \,\orcidlink{0000-0001-5196-1117}}
  \affiliation{\tubaf}
%    \email{ sven.groppe@informatik.tu-freiberg.de}

  \date{\today}

  \begin{abstract}
    Surface-code threshold estimates depend on the inference pipeline, including decoder and estimator choices. We compare decoders within a single LiDMaS+ workflow under Pauli-reference and digitized hybrid continuous-variable/discrete sweeps. In the Pauli-reference mode, the matching-style backend outperforms Union-Find and yields crossing median $p_c=0.0531$ (bootstrap interval $[0.0415,0.0572]$) and collapse fit $p_c=0.052$ ($\nu=1.35$). For the hybrid mode, a dense transition-window sweep at $d=3,5,7$ uses $\sigma\in[0.30,0.50]$ with step $0.01$ and $3000$ trials per point. After the initial exact-zero plateau is excluded from crossing localization, the matching-style backend gives interior crossing estimates $\sigma_c=0.4707$ for $(d=3,5)$ and $\sigma_c=0.3275$ for $(d=5,7)$; the latter lies in a low-LER region and remains estimator-sensitive. A targeted $d=9$ extension shows larger Union-Find LER at moderate-to-high $\sigma$ and matching-fallback rates up to $0.747$ at $\sigma=0.50$. In a $d=5$ neural-guidance sensitivity sweep, full learned reweighting reduces the sampled mean LER from $0.1773$ to $0.1663$ over $\sigma\in[0.35,0.55]$. These results show that estimator resolution and backend fallback diagnostics are part of an auditable decoder comparison.
  \end{abstract}

  \maketitle

  \section{Introduction}
  Threshold estimation asks whether increasing code distance reduces logical failure under a specified physical-noise model \cite{surfacecode1998}. The estimate is produced by a pipeline that includes noise generation, syndrome extraction, decoding, and finite-data analysis. It therefore depends on decoder and estimator choices.

  Hybrid continuous-variable/discrete settings motivated by GKP-style architectures \cite{gkp2001,tzitrin2021static,wayo2026lidmas} are a useful test case because continuous physical noise is digitized before decoding. We ask whether decoder choice changes threshold-oriented conclusions when the same protocol is applied to a Pauli-reference channel and a digitized CV-driven channel.

  LiDMaS+ provides the common workflow for this comparison. We evaluate a matching-style backend and Union-Find under matched grids, distances, and deterministic seeds, and include neural-guided matching in the hybrid sweeps. The scope is implementation-controlled: the Pauli-reference runs provide a stable comparison point, whereas the hybrid crossing estimates remain conditional on grid resolution, low-LER behavior, and backend fallbacks.

  A finite-distance curve crossing is useful only when its interpretation is stated carefully. It identifies a change in the sampled distance ordering under a specified decoder, estimator, and sweep grid; it is not, by itself, a decoder-independent material constant. The Pauli-reference mode supplies an internal control because its crossings can be compared with a collapse-based estimate under the same implementation. The hybrid mode tests a more limited question: whether a crossing remains identifiable after Gaussian displacement has been mapped to decoder-facing discrete events. We therefore distinguish interior crossings, boundary-valued coarse-grid proxies, and minimum-separation summaries throughout the study.

%  \subsection{Research Contributions}
  We address three questions: whether decoder choice changes threshold behavior after Gaussian displacement is digitized into effective Pauli faults in the GKP spirit \cite{gkp2001}; whether decoder ordering persists across $d=3,5,7$; and how finite-grid threshold summaries depend on the decoder and estimator \cite{pymatching2021,unionfind2017}.

  The study provides a reproducible LiDMaS+ protocol with matched seeds, grids, and outputs across matching-style, Union-Find, and neural-guided matching. Pauli-reference and hybrid CV-discrete results are reported together with decoder-failure and fallback diagnostics. Crossing estimates are interpreted with their sampled-grid resolution rather than as decoder-independent constants.

%  \subsection{Background and Related Work}
  Surface-code studies commonly report logical-error curves, curve crossings, and finite-size scaling summaries to evaluate whether larger distance improves reliability \cite{surfacecode1998}. A persistent comparability issue is that these summaries depend on the decoder and estimator stack used to produce them, even when the underlying code family is unchanged.

  Matching-based decoding is a standard surface-code reference because it combines decoding quality with practical tractability \cite{pymatching2021}. Union-Find is a lighter alternative whose approximation can alter curve shape and crossing stability \cite{unionfind2017}. We test whether that substitution changes the inferred threshold behavior under matched conditions.

  Neural-guided decoding introduces a separate modeling choice: learned weights can improve a matching heuristic while also creating new failure modes. We therefore report decoder-failure statistics alongside the neural-guided LER curves and describe the guidance model explicitly.

  GKP-style encodings and related photonic proposals motivate noise models in which continuous-variable disturbances are converted into digitized syndrome data and discrete correction actions \cite{gkp2001,tzitrin2021static,wayo2026lidmas}. We test whether decoder dependence persists after this CV-to-discrete mapping.

  Related decoder and code-design work includes XZZX-oriented and bias-tailored strategies, scalable parallel decoders, and gate-oriented fault-tolerant constructions \cite{bonilla2021xzzx,darmawan2021xzzxkerrcat,higgott2023improved,skoric2023parallel,tan2023scalable,xu2023tailored,dua2024clifford,kobayashi2024crosscap,sahay2025transversal,kang2024joint,ha2025architectures}.

  Experimental demonstrations across superconducting, silicon-spin, and bosonic platforms provide complementary evidence on logical protection and below-threshold operation \cite{zhao2022surfacecode,krinner2022repeated,takeda2022silicon,ni2023breakeven,acharya2023suppressing,reglade2024catcontrol,ding2025kerrcat,google2025belowthreshold}. Related work addresses logical-noise mitigation and high-threshold memory constructions \cite{smith2024mitigating,bravyi2024ldpcmemory,xu2023squeezedcat}.

  \subsection{LiDMaS+ as the Experimental Platform}
  LiDMaS+ is the experimental platform for this study, not the subject of a software benchmark. All runs use the \texttt{--surface\_threshold} workflow with explicit decoder, distance, sweep-window, trial-count, and output arguments. The common code path removes cross-tool differences from the decoder comparison.

  The matching-style, Union-Find, and neural-guided matching backends use the same LiDMaS+ decoder interface. Study-specific shell scripts call the compiled binary for fixed-distance baselines, multi-distance sweeps, and threshold-oriented runs without changing the surrounding experiment logic.

  Each run records the seed, decoder, distance, physical-noise parameter, trial count, logical-error statistics, confidence intervals, and decoder diagnostics in machine-readable outputs. Pauli-reference runs additionally emit scaling reports and JSON summaries for finite-size analysis. Hybrid crossing estimates are calculated from the saved sampled curves.

  The dedicated \texttt{examples/paper\_runs/paper\_01} workflow standardizes the experiment matrix, output directories, and merged tabular-data generation without altering decoder logic.

  This study does not benchmark runtime or accuracy against external platforms. Stim is a dedicated QEC simulator; Qiskit, PennyLane, and Qibo are general-purpose frameworks that can support QEC workflows; and Mitiq is an error-mitigation toolkit. Cross-platform performance comparisons are outside the present scope.

  \section{Methodology}
  For each operating point, LiDMaS+ samples noise, extracts syndrome data, runs one decoder backend, applies the correction, and records logical failure. Fixed seeds, sweep grids, and trial budgets provide matched decoder comparisons. The pipeline then computes LER confidence intervals, aggregates distance-dependent curves, and runs mode-specific threshold analyses.

  Figure~\ref{fig:surface-threshold-workflow} summarizes the matched simulation and analysis workflow.

  \begin{figure*}[t]
    \centering
    \includegraphics[width=0.98\textwidth]{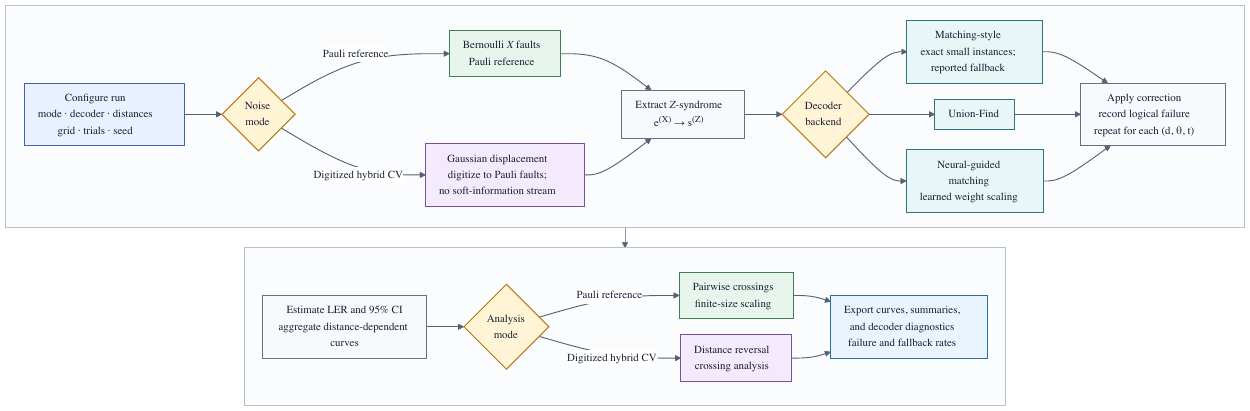}
    \caption{LiDMaS+ workflow used for the controlled decoder comparison. Pauli-reference and digitized hybrid CV modes share the syndrome-extraction, decoder, correction, aggregation, and reporting paths. The hybrid mode passes parity-mapped Pauli faults, not continuous soft information, to the decoder backend.}
    \label{fig:surface-threshold-workflow}
  \end{figure*}

  \subsection{Noise Models}
  The study uses two LiDMaS+ modes. The Pauli-reference mode has sweep variable $p$ and samples a code-capacity-style surrogate, not a noisy-gate circuit-level schedule. It provides an internal decoder reference.

  The hybrid CV-discrete mode has sweep variable $\sigma$. Each trial samples Gaussian displacement and digitizes it to Pauli components in a GKP-inspired manner. The threshold harness decodes one syndrome sector and does not pass continuous soft information to the decoder backend. The resulting curves describe a digitized CV-driven effective channel, not analog-information decoding.

  The paired modes serve different interpretive roles. The Pauli-reference mode isolates decoder behavior under a direct Bernoulli fault parameter $p$. The hybrid mode introduces a preceding displacement-to-Pauli map and uses $\sigma$ as the physical sweep parameter. Consequently, numerical values of $p$ and $\sigma$ should not be compared as if they were interchangeable error probabilities. The comparison of interest is within each mode: whether distance ordering, decoder ordering, and crossing stability change when the decoder backend is varied while the surrounding workflow is held fixed.

  The lack of continuous soft information is also an explicit scope boundary. Digitization preserves the discrete parity events consumed by the current surface-code decoder interface, but it discards confidence information that could be used by an analog-aware decoder. The hybrid curves therefore measure decoder dependence after digitization. They do not estimate the improvement that might be obtained from continuous likelihoods, nor do they constitute a benchmark of a full circuit-level GKP architecture. This distinction is important when interpreting the hybrid crossing values and when comparing them with thresholds reported for other noise models.

  Operational definitions for both modes are summarized in Table~\ref{tab:mode-definitions}.

  \begin{table*}[t]
    \centering
    \small
    \caption{Operational mode definitions used in this study.}
    \label{tab:mode-definitions}
    \begin{tabular}{lllll}
      \toprule
      Mode & Sweep & Fault map & Decoded branch & Soft info \\
      \midrule
      Pauli reference & $p$ & Bernoulli $X$ & $e^{(X)} \rightarrow s^{(Z)}$ & none \\
      Hybrid (digitized CV) & $\sigma$ & Gaussian $\rightarrow$ parity-mapped Pauli & $e^{(X)} \rightarrow s^{(Z)}$ & none (no continuous stream) \\
      \bottomrule
    \end{tabular}
  \end{table*}

  \subsection{Decoders Compared}
  The comparison includes the LiDMaS+ matching-style backend, exposed by the \texttt{mwpm} run option, Union-Find, and neural-guided matching. Matching-based decoding is a standard surface-code reference \cite{pymatching2021}; the present backend is implementation-specific and is reported with fallback diagnostics.

  Union-Find is the approximate alternative \cite{unionfind2017}. Its inclusion tests whether a lighter decoder preserves the logical-error trends and crossing summaries obtained with the matching-style backend.

  For small defect sets, the LiDMaS+ matching-style backend solves the pair/boundary matching objective exactly with a subset dynamic program. Above the configured small-instance limit, it uses a greedy fallback. The hybrid transition-window and $d=9$ runs report the fallback rate. These results are not benchmarks of a production Blossom implementation.

  Neural-guided matching reweights the same matching graph; it is not an end-to-end neural decoder. A trained linear model maps defect-coordinate offsets and boundary-proximity features to multiplicative edge-weight adjustments before matching. Training minimizes a simulator-in-the-loop objective over hybrid points with a decoder-failure penalty. Hybrid results report both decoder-failure and matching-fallback diagnostics.

  The comparison is therefore between concrete LiDMaS+ backends, not between idealized decoder families. For the matching-style backend, the fallback rate marks the fraction of requests that leave the exact small-instance solver and enter the greedy path. A low fallback rate supports interpretation in terms of the exact objective implemented by the backend. A high rate means that the reported LER increasingly reflects the fallback policy as well. For neural-guided matching, the decoder-failure diagnostic is separate from logical failure: it records cases in which the guided correction fails the syndrome-consistency check. Reporting both quantities prevents an apparently favorable LER curve from concealing backend instability.

  \subsection{Threshold Estimation Protocol}
  For each decoder, distance, and physical-noise value, LiDMaS+ generates repeated samples, extracts syndromes, decodes them, and records logical failures. The LER is the empirical failure frequency. Confidence intervals are reported for each point, and each decoder uses the same noise grid and trial count.

  For the Pauli-reference mode, built-in threshold-estimation and finite-size-scaling routines return pairwise crossings, a crossing median, and a collapse-based critical-point estimate with bootstrap uncertainty. For the hybrid mode, matched multi-distance sweeps support distance-reversal and crossing analysis, but not an asymptotic threshold claim.

  Crossing localization is performed on the sampled distance-dependent curves. An interior sign change in the LER difference between two distances supports interpolation between adjacent grid points. In the hybrid mode, the estimator first excludes an initial exact-zero plateau, because equality caused by zero observed failures does not localize a transition. The exclusion does not remove nonzero observations from the transition region. If no sign change remains on the sampled grid, the reported minimum-separation point is labeled as such rather than presented as a crossing. The collapse fit used for the Pauli-reference mode is a complementary finite-size summary; agreement with crossing estimates is evidence of stability within the sampled protocol, not a universality claim.

  \subsection{Algorithmic Summary}
  Algorithms~\ref{alg:surface-threshold}, \ref{alg:mwpm-decoder}, \ref{alg:uf-decoder}, and \ref{alg:neural-mwpm} summarize the threshold loop and the three decoders.

  \begin{algorithm}[t]
\caption{Surface-code threshold run (generic decoder)}
\label{alg:surface-threshold}
\KwInput{Decoder $D$, mode $m$, distances $\mathcal{D}$, sweep grid $\mathcal{S}$, trials $T$, base seed $s_0$}
\KwOutput{Logical-error estimates and confidence intervals for each $(d,\theta)$ point}
\ForEach{$d \in \mathcal{D}$}{
  \ForEach{$\theta \in \mathcal{S}$}{
    failures $\gets 0$\;
    \ForLoop{$t=1$ \KwTo $T$}{
      Set RNG seed from $(s_0,d,\theta,t)$\;
      Sample physical noise under $(m,\theta)$ and extract syndrome $y_t$\;
      $c_t \gets \mathrm{Decode}_D(y_t,d)$\;
      \If{$c_t$ yields a logical failure}{
        failures $\gets$ failures $+ 1$\;
      }
    }
    Record $\mathrm{LER}(d,\theta)=\mathrm{failures}/T$ and CI\;
  }
}
  \end{algorithm}

  \begin{algorithm}[t]
\caption{Matching-style surface decoder}
\label{alg:mwpm-decoder}
\KwInput{Syndrome graph $G=(V,E)$, boundary set $\partial V$, edge weights $w$}
\KwOutput{Surface correction $c$}
Build matching instance with pair and boundary edges\;
Assign edge and boundary costs from $w$\;
$M \gets \mathrm{MatchingSolveWithFallback}(G,w)$\;
Initialize correction bitmask $c \gets 0$\;
\ForEach{matched pair $(u,v)\in M$}{
  Route minimum-weight path $P(u,v)$\;
  Toggle qubits along $P(u,v)$ in $c$\;
}
Verify syndrome consistency $Hc = y \pmod 2$\;
\textbf{return} $c$\;
  \end{algorithm}

  \begin{algorithm}[t]
\caption{Union-Find surface decoder}
\label{alg:uf-decoder}
\KwInput{Defect set $\mathcal{X}$ on lattice $L$}
\KwOutput{Surface correction $c$}
Initialize one odd cluster per defect\;
\While{an odd cluster exists}{
  Grow each odd-cluster frontier by one edge\;
  Merge clusters on frontier collisions\;
  Attach a cluster to boundary when reached\;
}
Build spanning forest in each merged cluster\;
Peel leaves to satisfy parity and mark correction flips in $c$\;
Verify syndrome consistency\;
\textbf{return} $c$\;
  \end{algorithm}

  \begin{algorithm}[t]
\caption{Neural-guided matching surface decoder}
\label{alg:neural-mwpm}
\KwInput{Defects $\mathcal{X}$, trained model $f_{\phi}$, base matching weights $w$}
\KwOutput{Guided correction $c$ and decoder-failure diagnostic if needed}
\ForEach{defect pair $(i,j)$}{
  Extract features $x_{ij}$ (distance and boundary-aware terms)\;
  $\alpha_{ij} \gets \mathrm{Clip}(f_{\phi}(x_{ij}))$\;
  $\tilde{w}_{ij} \gets w_{ij}+\lambda(\alpha_{ij}w_{ij}-w_{ij})$, $\lambda\in[0,1]$\;
}
\ForEach{defect $i$}{
  Compute guided boundary weight $\tilde{w}_{i,\partial}$\;
}
Solve matching instance on guided weights $\tilde{w}$ and lift to correction $c$\;
\If{syndrome verification fails}{
  Emit decoder-failure diagnostic\;
}
\textbf{return} $c$\;
  \end{algorithm}

  \subsection{Supplementary Decoder Trace}
  Appendix~\ref{app:indexed-trace} provides the indexed $d=5$ lattice diagram and single-instance decoder trace.

  \subsection{Metrics}
  The primary metric is the logical error rate (LER). Main comparison figures display 95\% Wilson-score uncertainty bands. The Pauli-reference threshold analysis also reports crossing estimates, collapse-based critical-point estimates, fitted exponent $\nu$, and collapse cost. The hybrid analysis reports decoder-labeled crossing estimates or proxies from the sampled multi-distance curves and identifies grid-limited cases.

  The Wilson intervals quantify finite-sample uncertainty in the observed failure frequency at each operating point. They do not include systematic uncertainty from the noise map, decoder implementation, sweep window, or crossing estimator. The AUC proxy is used as a compact curve-level summary over a fixed sampled window; it is not a threshold estimate and should not be compared across different windows without qualification. Decoder-failure and fallback rates are reported alongside LER because they identify implementation behavior that an LER value alone cannot expose.

  \section{Experimental Design}
  The experiment matrix in Table~\ref{tab:experimental-matrix} covers fixed-distance baselines, multi-distance sweeps, threshold-oriented runs, a dense hybrid transition window, a targeted $d=9$ extension, and a neural-guidance sensitivity sweep. Within each run block, decoders use deterministic seeds, identical sweep grids, and the same tabular output format.

  \begin{table*}[t]
    \centering
    \small
    \caption{Experimental matrix for the decoder-comparison study (seed 1337; common LiDMaS+ workflow).}
    \label{tab:experimental-matrix}

    \begin{tabular}{llllll}
      \toprule
      Run & Mode & $d$ & Sweep & Trials & Purpose \\
      \midrule

      Pauli baseline &
      Pauli &
      5 &
      $p\in[0.03,0.12]$, 0.01 &
      3000 &
      Decoder comparison (fixed $d$) \\

      Hybrid baseline &
      Hybrid (digitized CV) &
      5 &
      $\sigma\in[0.05,0.60]$, 0.05 &
      2000 &
      Decoder comparison (fixed $d$)  \\

      Hybrid multi-$d$ &
      Hybrid (digitized CV) &
      3,5,7 &
      $\sigma\in[0.05,0.60]$, 0.05 &
      2000 &
      Distance scaling trends \\

      Pauli threshold &
      Pauli &
      3,5,7 &
      $p\in[0.04,0.12]$, 0.01 &
      2000 &
      Crossings / FSS \\

      Hybrid threshold &
      Hybrid (digitized CV) &
      3,5,7 &
      $\sigma\in[0.05,0.60]$, 0.05 &
      2000 &
      Crossing proxy; reversal \\

      Dense transition-window hybrid &
      Hybrid (digitized CV) &
      3,5,7 &
      $\sigma\in[0.30,0.50]$, 0.01 &
      3000 &
      Interior crossing localization \\

      Targeted distance extension &
      Hybrid (digitized CV) &
      9 &
      $\sigma\in[0.35,0.50]$, 0.05 &
      1000 &
      Scaling check; fallback diagnostics \\

      Neural sensitivity &
      Hybrid (digitized CV) &
      5 &
      $\sigma\in[0.35,0.55]$, 0.05 &
      3000 &
      Guidance strength $\lambda=0,0.5,1$ \\

      \bottomrule
    \end{tabular}
  \end{table*}

  The fixed-distance Pauli-reference baseline uses $d=5$, $p\in[0.03,0.12]$ with step $0.01$, and $3000$ trials per point. It checks whether the workflow recovers the expected advantage of matching-style decoding over Union-Find in the discrete reference mode.

  The fixed-distance hybrid baseline also uses $d=5$, with $\sigma\in[0.05,0.60]$, step $0.05$, and $2000$ trials per point for matching-style, Union-Find, and neural-guided matching. It tests whether decoder separation remains visible after CV displacement is digitized into decoder-facing events. No continuous soft-information stream is passed to the decoder backend.

  The hybrid multi-distance sweep uses $d=3,5,7$, the same $\sigma$ window, and the same three decoders. It identifies changes in distance ordering as noise increases and supplies the curves used for crossing analysis.

  The threshold-oriented Pauli-reference run uses $d=3,5,7$, $p\in[0.04,0.12]$ with step $0.01$, $2000$ trials per point, and $100$ bootstrap samples in the saved scaling report. It produces pairwise crossings, a crossing median, and collapse-fit summaries for matching-style and Union-Find decoding. The corresponding hybrid run uses $d=3,5,7$, $\sigma\in[0.05,0.60]$ with step $0.05$, and $2000$ trials per point. Its coarse grid supports distance-reversal analysis and a crossing proxy, not a precise asymptotic threshold.

  Three targeted runs supplement the coarse hybrid grid. A $d=3,5,7$ transition-window sweep uses $\sigma\in[0.30,0.50]$ with step $0.01$ and $3000$ trials per point. A $d=9$ extension samples $\sigma\in[0.35,0.50]$ with step $0.05$ and $1000$ trials per point. A $d=5$ sensitivity sweep compares guidance strengths $\lambda=0,0.5,1$ over $\sigma\in[0.35,0.55]$ with step $0.05$ and $3000$ trials per point.

  Matched run settings reduce avoidable sources of variation between decoders. Each decoder receives the same distance set, physical-noise grid, trial budget, and deterministic seed schedule within a run block. This design does not remove Monte Carlo uncertainty, and it does not force two decoders to follow the same correction path. It ensures that differences in the reported curves are evaluated under the same sampled protocol. The targeted runs then spend additional trials where the coarse sweep indicates that interpretation is most sensitive to grid resolution or backend behavior.

  \section{Results}
  \begin{figure*}[t]
    \centering
    \subfloat[Pauli baseline at $d=5$.]{
      \includegraphics[width=0.48\textwidth]{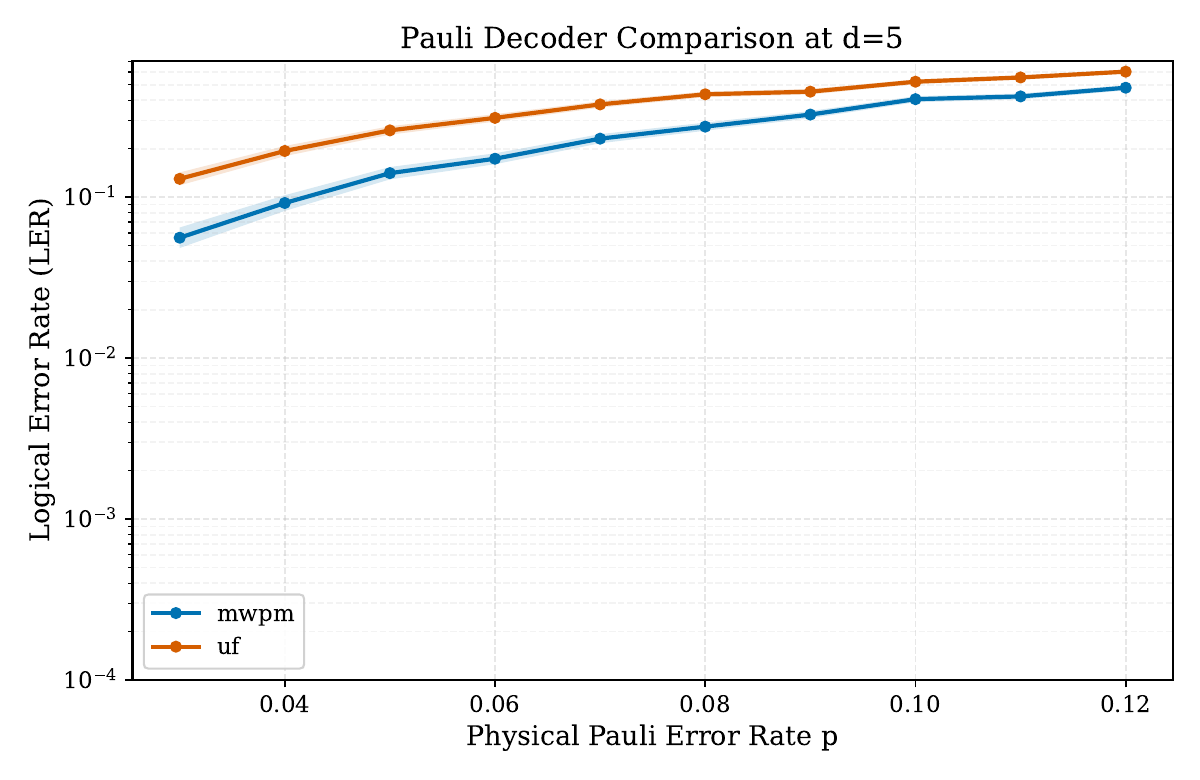}
    }
    \hfill
    \subfloat[Hybrid baseline at $d=5$.]{
      \includegraphics[width=0.48\textwidth]{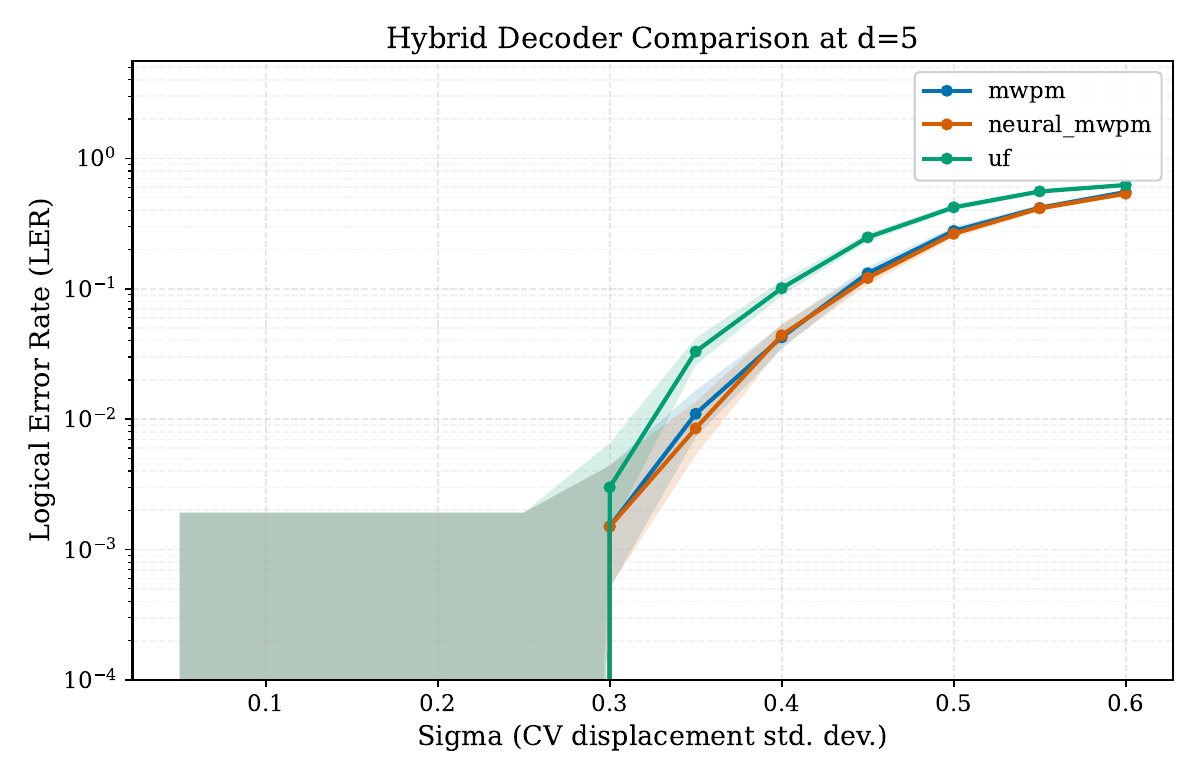}
    }
    \caption{Fixed-distance decoder comparison under matched sweep settings. Panel (a) compares matching-style and Union-Find decoding in the Pauli-reference mode. Panel (b) compares matching-style, Union-Find, and neural-guided matching in the hybrid mode at $d=5$.}
    \label{fig:baseline-comparison}
  \end{figure*}

  \subsection{Baseline Decoder Comparison}
  In the Pauli-reference baseline at $d=5$ [Fig.~\ref{fig:baseline-comparison}(a)], the matching-style backend has lower LER than Union-Find at every sampled point. Representative values are $0.056$ versus $0.130$ at $p=0.03$, $0.141$ versus $0.260$ at $p=0.05$, $0.274$ versus $0.436$ at $p=0.08$, and $0.479$ versus $0.602$ at $p=0.12$. The sampled mean LER is $0.260$ for matching-style decoding and $0.384$ for Union-Find; the corresponding AUC proxies are $0.0233$ and $0.0347$.

  The reported confidence intervals remain separated at representative low-, mid-, and high-noise points. The Pauli-reference result provides the comparison point for the hybrid analysis.

  In the hybrid baseline at $d=5$ [Fig.~\ref{fig:baseline-comparison}(b)], matching-style and neural-guided matching remain close, while Union-Find rises faster with $\sigma$. Representative values for matching-style, Union-Find, and neural-guided matching are $(0.0425,0.101,0.044)$ at $\sigma=0.40$, $(0.278,0.421,0.264)$ at $\sigma=0.50$, and $(0.552,0.6245,0.5365)$ at $\sigma=0.60$. Table~\ref{tab:hybrid-baseline-summary} reports sampled mean LER $(0.1195,0.1657,0.1158)$ and AUC proxies $(0.0579,0.0838,0.0561)$ in the same decoder order.

  The fixed-distance baselines establish two points for the multi-distance analysis. First, decoder separation is visible before threshold estimation is applied, so it is not created solely by a crossing heuristic. Second, the close agreement between matching-style and neural-guided matching at $d=5$ limits the strength of any learned-decoder claim: the sensitivity analysis must determine whether guidance changes the sampled curve consistently, rather than relying on a single operating point.

  \begin{table*}[t]
    \centering
    \caption{Hybrid fixed-distance summary at $d=5$ for $\sigma\in[0.05,0.60]$ (2000 trials/point).}
    \label{tab:hybrid-baseline-summary}
    \begin{tabular}{lcccc}
      \toprule
      Decoder & Points & Mean LER & AUC proxy & Max decoder fail rate \\
      \midrule
      Matching-style & 12 & 0.1195 & 0.0579 & 0.0000 \\
      Union-Find & 12 & 0.1657 & 0.0838 & 0.0000 \\
      Neural-guided matching & 12 & 0.1158 & 0.0561 & 0.0010 \\
      \bottomrule
    \end{tabular}
  \end{table*}

  \begin{figure*}[t]
    \centering
    \subfloat[Matching-style backend.]{
      \includegraphics[width=0.32\textwidth]{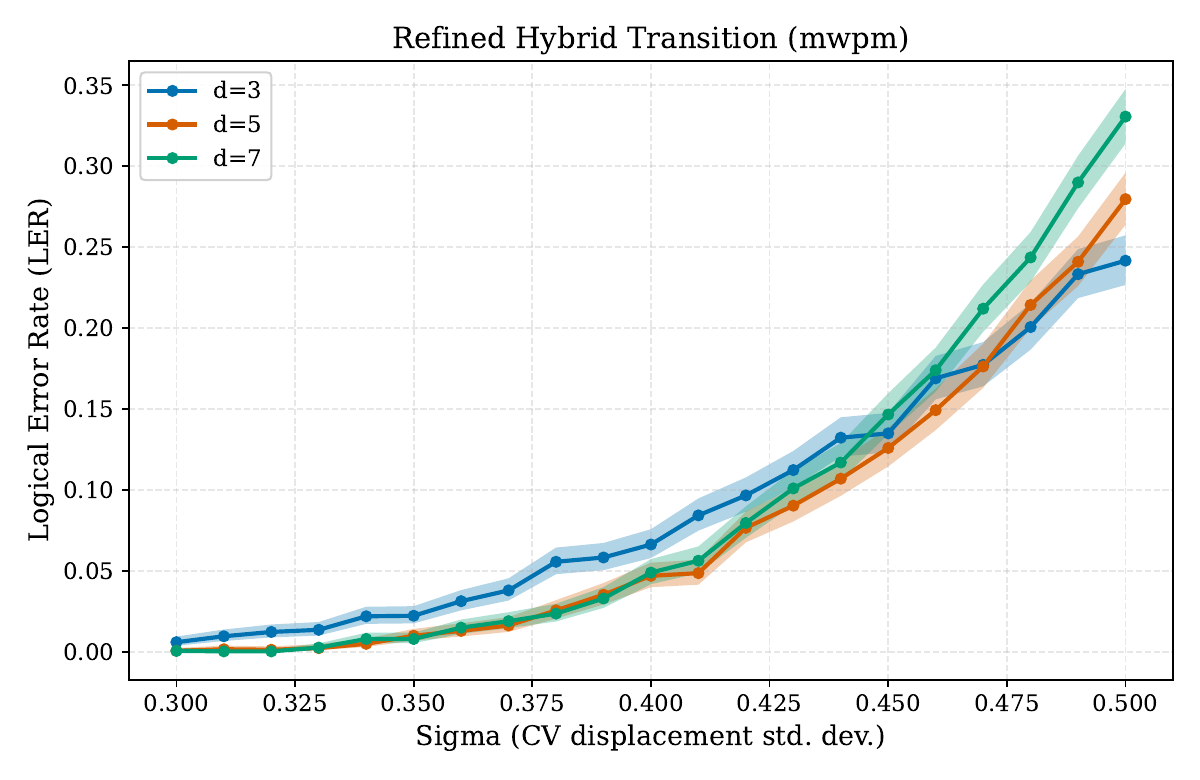}
    }
    \hfill
    \subfloat[Union-Find.]{
      \includegraphics[width=0.32\textwidth]{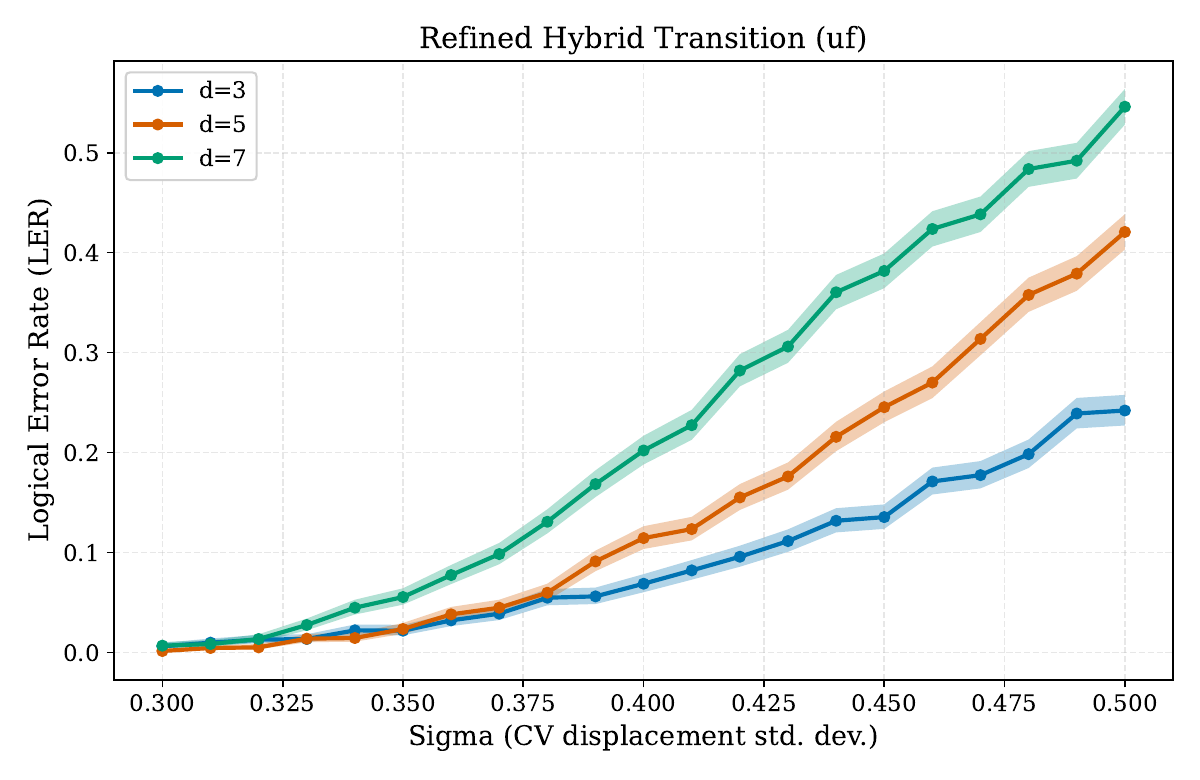}
    }
    \hfill
    \subfloat[Neural-guided matching.]{
      \includegraphics[width=0.32\textwidth]{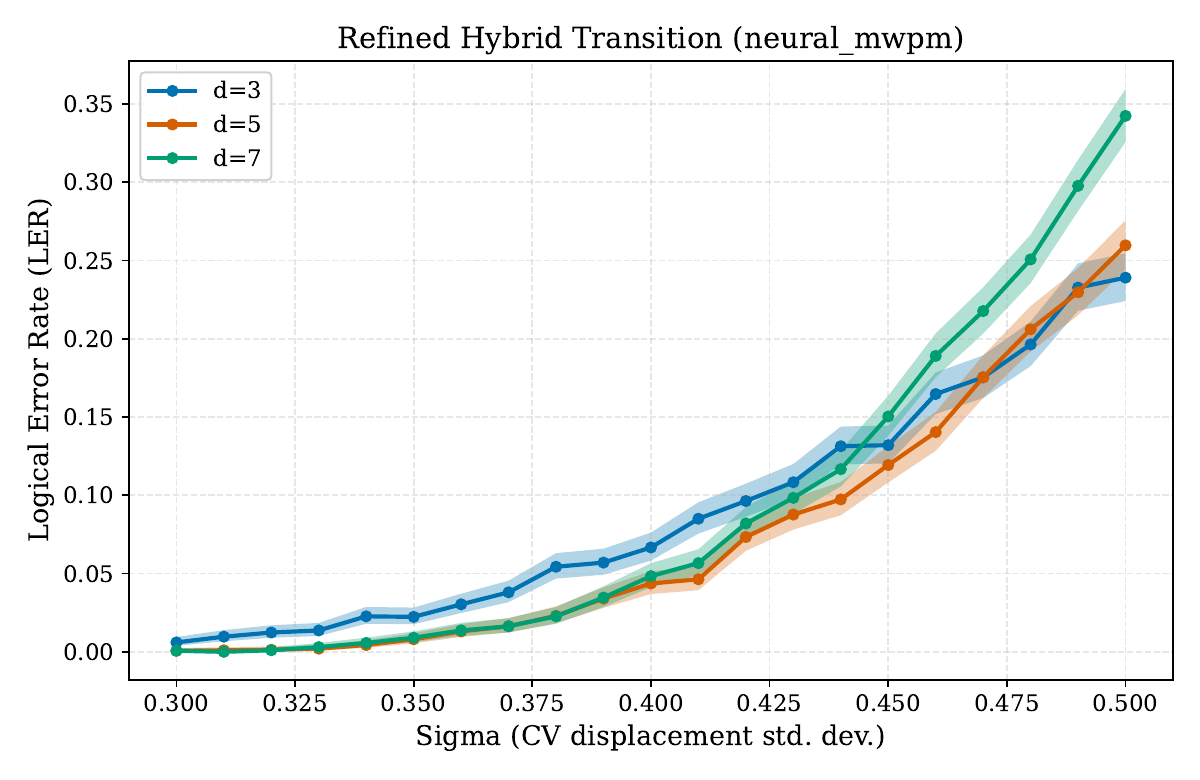}
    }
    \caption{Dense hybrid transition-window sweeps for $d=3,5,7$ on $\sigma\in[0.30,0.50]$ with step $0.01$ and $3000$ trials per point. Shaded regions are 95\% Wilson-score confidence bands.}
    \label{fig:transition-window-hybrid-crossing}
  \end{figure*}

  \begin{figure}[t]
    \centering
    \includegraphics[width=\columnwidth]{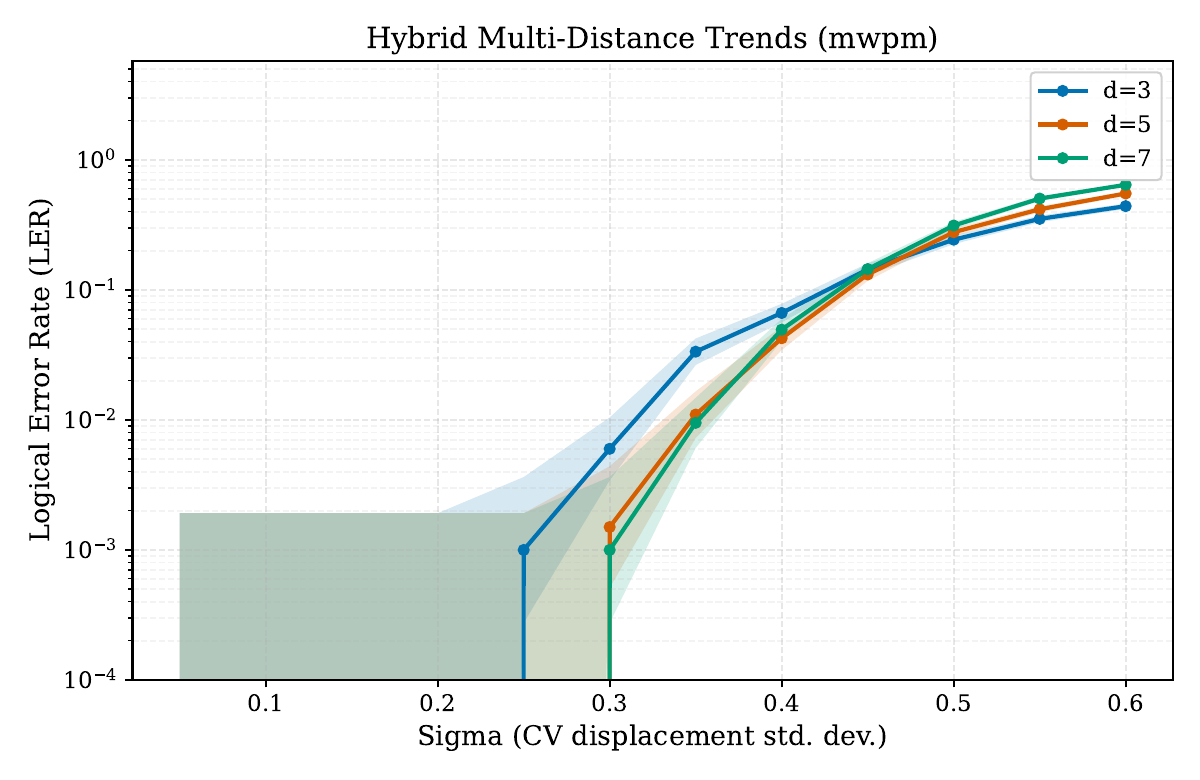}
    \caption{Hybrid multi-distance logical-error curves for $d=3,5,7$ using the matching-style backend. Distance ordering reverses between the low- and high-$\sigma$ regimes; Table~\ref{tab:hybrid-multidistance-summary} summarizes the decoder comparison.}
    \label{fig:hybrid-multidistance}
  \end{figure}

  \subsection{Distance-Dependent Threshold Trends}
  The multi-distance hybrid experiment extends the comparison to $d=3,5,7$ [Fig.~\ref{fig:hybrid-multidistance} and Table~\ref{tab:hybrid-multidistance-summary}]. At $d=3$, the AUC proxies are close: $0.0534$ for matching-style decoding, $0.0527$ for Union-Find, and $0.0523$ for neural-guided matching. At $d=5$ and $d=7$, the Union-Find AUC increases to $0.0838$ and $0.1085$. Neural-guided matching remains close to matching-style decoding in AUC but has higher decoder-failure diagnostics at high noise for $d=7$.

  For the matching-style backend, representative LER values at $\sigma=0.30$ are $(0.006,0.0015,0.001)$ for $d=(3,5,7)$ and at $\sigma=0.40$ are $(0.0665,0.0425,0.0495)$. The ordering reverses by $\sigma=0.50$, where the values are $(0.2435,0.278,0.313)$, and widens at $\sigma=0.60$ to $(0.441,0.552,0.644)$. Union-Find shows the same qualitative reversal with larger high-noise LER; for $d=7$, the values are $0.2035$ at $\sigma=0.40$ and $0.6935$ at $\sigma=0.60$. Neural-guided matching tracks matching-style decoding through moderate noise and develops higher decoder-failure diagnostics at the largest sampled $\sigma$.

  Decoder dependence appears differently in the two modes. Matching-style decoding outperforms Union-Find throughout the Pauli-reference sweep. In the hybrid mode, Union-Find degradation is concentrated at larger distances and moderate-to-high $\sigma$. Neural-guided matching remains close to matching-style decoding in LER but develops decoder-failure diagnostics at the largest sampled noise values. These data support decoder-dependent hybrid behavior, not an asymptotic threshold claim.

  \begin{table*}[t]
    \centering
    \caption{Hybrid multi-distance summary from sampled curves (2000 trials/point).}
    \label{tab:hybrid-multidistance-summary}
    \begin{tabular}{lccc}
      \toprule
      Decoder & AUC at $d=3$ & AUC at $d=5$ & AUC at $d=7$ \\
      \midrule
      Matching-style & 0.0534 & 0.0579 & 0.0672 \\
      Union-Find & 0.0527 & 0.0838 & 0.1085 \\
      Neural-guided matching & 0.0523 & 0.0561 & 0.0695 \\
      \bottomrule
    \end{tabular}
  \end{table*}

  \begin{table*}[t]
    \centering
    \caption{Pauli threshold and finite-size scaling summary extracted from the $d=3,5,7$ runs.}
    \label{tab:pauli-threshold-summary}
    \begin{tabular}{lccccc}
      \toprule
      Decoder & Crossing median $p_c$ & Crossing interval & Collapse $p_c$ & $\nu$ & Collapse cost \\
      \midrule
      Matching-style & 0.0531 & $[0.0415,\,0.0572]$ & 0.0520 & 1.35 & $8.05\times 10^{-5}$ \\
      Union-Find & N/A & -- & 0.0400 & 0.50 & $2.33\times 10^{-3}$ \\
      \bottomrule
    \end{tabular}
  \end{table*}

  \begin{figure*}[t]
    \centering
    \subfloat[Matching-style Pauli-reference threshold dataset ($d=3,5,7$).]{
      \includegraphics[width=0.48\textwidth]{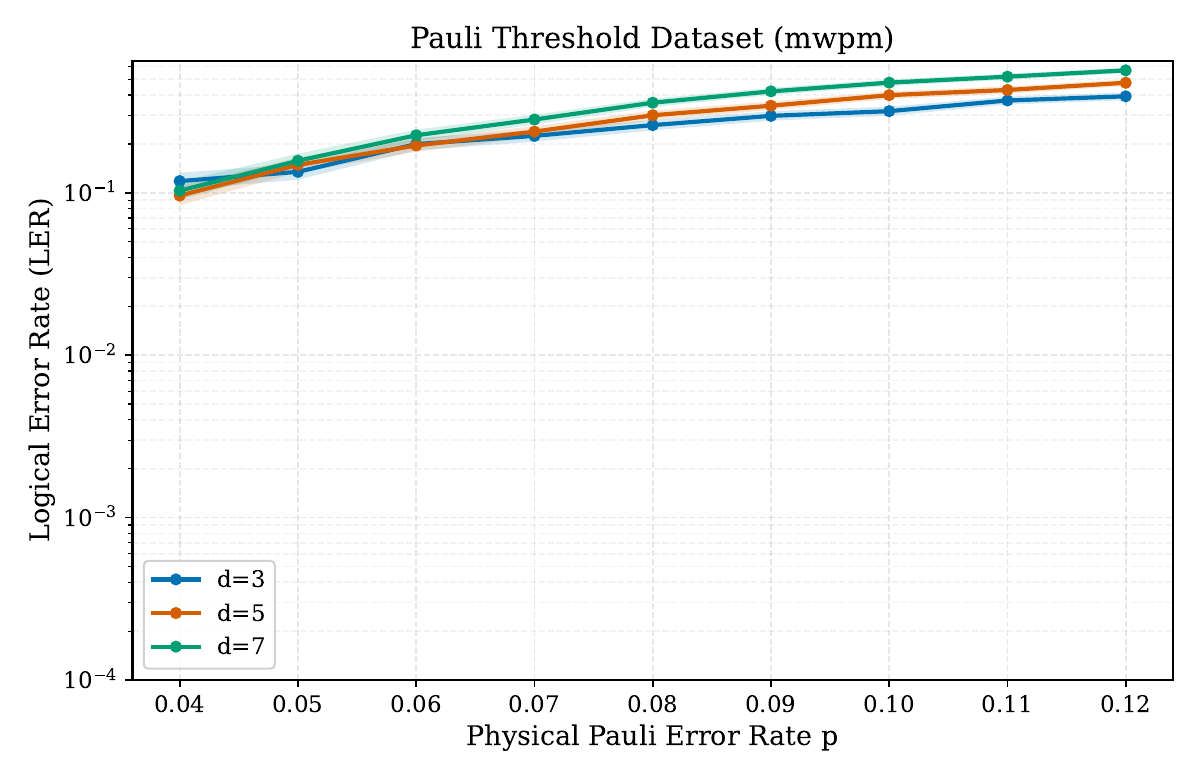}
    }
    \hfill
    \subfloat[Union-Find Pauli-reference threshold dataset ($d=3,5,7$).]{
      \includegraphics[width=0.48\textwidth]{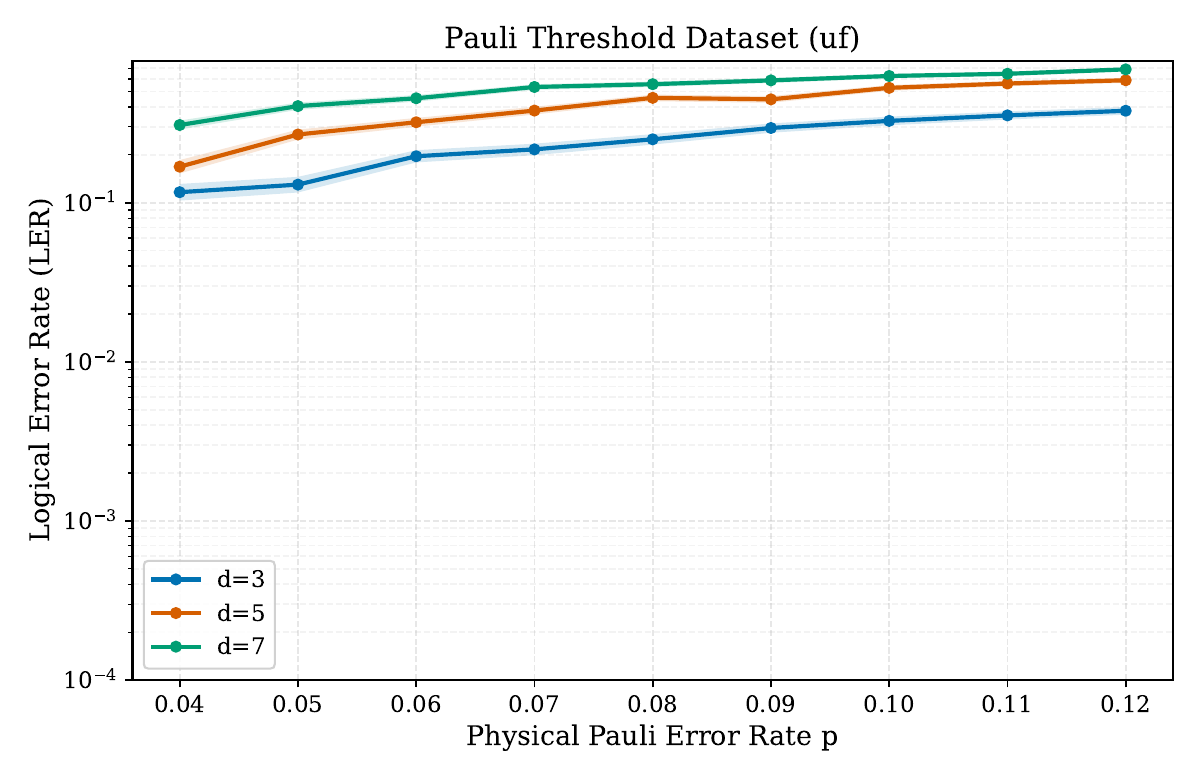}
    }
    \caption{Pauli-reference threshold curves for matching-style and Union-Find decoding.}
    \label{fig:pauli-threshold-curves}
  \end{figure*}

  \begin{table*}[t]
    \centering
    \caption{Dense transition-window hybrid crossing summary after excluding each initial exact-zero low-noise plateau.}
    \label{tab:hybrid-crossing-summary}
    \begin{tabular}{lcccc}
      \toprule
      Decoder & $\sigma_{c}(d=3,5)$ & Method & $\sigma_{c}(d=5,7)$ & Method \\
      \midrule
      Matching-style & 0.4707 & linear interp. & 0.3275 & linear interp. \\
      Union-Find & 0.3296 & linear interp. & 0.3100 & min.\ $|\Delta\mathrm{LER}|$ \\
      Neural-guided matching & 0.4700 & grid exact & 0.3225 & linear interp. \\
      \bottomrule
    \end{tabular}
  \end{table*}

  \begin{figure}[t]
    \centering
    \includegraphics[width=\columnwidth]{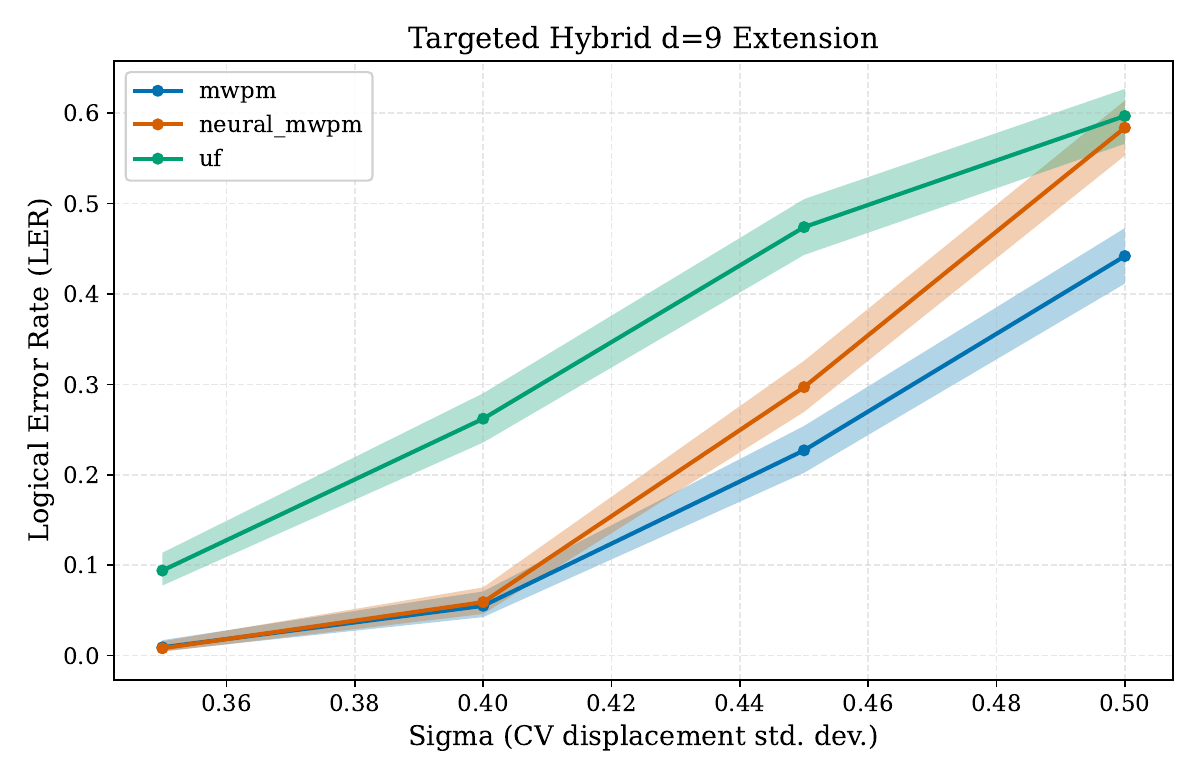}
    \caption{Targeted $d=9$ hybrid extension on $\sigma\in[0.35,0.50]$ with step $0.05$ and $1000$ trials per point. Shaded regions are 95\% Wilson-score confidence bands.}
    \label{fig:hybrid-d9-extension}
  \end{figure}

  \begin{table}[t]
    \centering
    \small
    \caption{Targeted $d=9$ hybrid diagnostics.}
    \label{tab:hybrid-d9-summary}
    \begin{tabular}{lccc}
      \toprule
      Decoder & Mean LER & Max fail & Max fallback \\
      \midrule
      Matching-style & 0.1833 & 0 & 0.747 \\
      Union-Find & 0.3568 & 0 & 0 \\
      Neural-guided matching & 0.2370 & 0.196 & 0.747 \\
      \bottomrule
    \end{tabular}
  \end{table}

  \begin{figure}[t]
    \centering
    \includegraphics[width=\columnwidth]{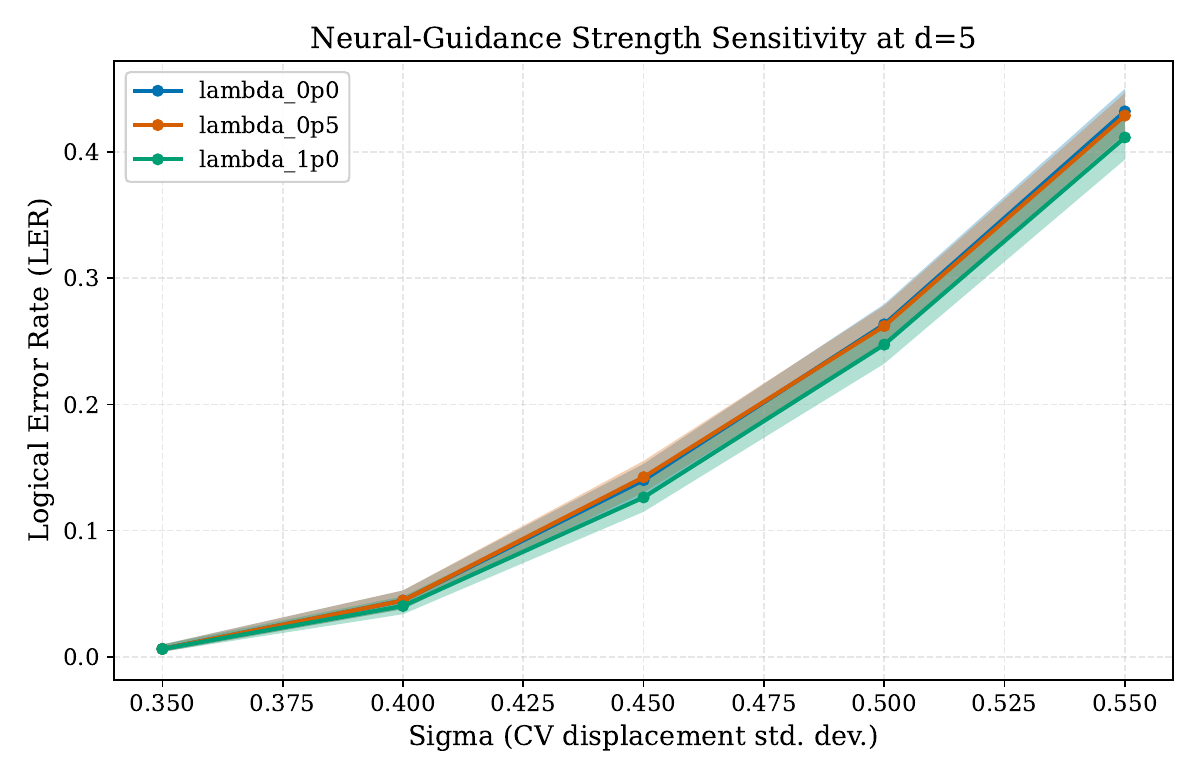}
    \caption{Neural-guidance sensitivity at $d=5$ for $\lambda=0,0.5,1$ over $\sigma\in[0.35,0.55]$. Shaded regions are 95\% Wilson-score confidence bands.}
    \label{fig:neural-guidance-sensitivity}
  \end{figure}

  \begin{table}[t]
    \centering
    \small
    \caption{Neural-guidance sensitivity at $d=5$ (3000 trials/point).}
    \label{tab:neural-guidance-summary}
    \begin{tabular}{lcc}
      \toprule
      Guidance strength & Mean LER & Max LER \\
      \midrule
      $\lambda=0$ & 0.1773 & 0.4320 \\
      $\lambda=0.5$ & 0.1768 & 0.4287 \\
      $\lambda=1$ & 0.1663 & 0.4113 \\
      \bottomrule
    \end{tabular}
  \end{table}

  \subsection{Crossing and Threshold Estimates}
  Table~\ref{tab:pauli-threshold-summary} and Fig.~\ref{fig:pauli-threshold-curves} report the Pauli-reference threshold analysis. For the matching-style backend, pairwise crossings occur at $p_c\approx 0.0622$ for $(d=3,d=5)$ and $p_c\approx 0.0439$ for $(d=3,d=7)$. The crossing median is $0.0531$ with bootstrap interval $[0.0415,\,0.0572]$. The collapse fit gives $p_c=0.052$ with interval $[0.040,\,0.064]$, exponent $\nu=1.35$, and collapse cost $8.05\times 10^{-5}$.

  For Union-Find, no crossings are detected on the sampled Pauli-reference grid. The collapse fit reaches the lower boundary of the search region, with $p_c=0.040$, $\nu=0.5$, and collapse cost $2.33\times 10^{-3}$. Under this protocol, Union-Find does not yield a stable scalar threshold summary.

  Figure~\ref{fig:transition-window-hybrid-crossing} and Table~\ref{tab:hybrid-crossing-summary} report the dense hybrid transition-window analysis. The estimator excludes the initial exact-zero plateau before searching for interior sign changes. Matching-style decoding gives $\sigma_c=0.4707$ for $(d=3,d=5)$ and $\sigma_c=0.3275$ for $(d=5,d=7)$. Neural-guided matching gives $0.4700$ and $0.3225$. The $(d=5,d=7)$ estimates occur at low LER and remain estimator-sensitive. Union-Find gives an interior $(d=3,d=5)$ estimate of $0.3296$; the $(d=5,d=7)$ pair has no sign change on the dense grid and is summarized by its minimum absolute curve separation at $\sigma=0.3100$.

  The $d=9$ extension [Fig.~\ref{fig:hybrid-d9-extension} and Table~\ref{tab:hybrid-d9-summary}] gives mean LER $0.3568$ for Union-Find and $0.1833$ for the matching-style backend over $\sigma\in[0.35,0.50]$. The matching-style backend reaches a maximum greedy-fallback rate of $0.747$. Neural-guided matching reaches the same fallback rate and a decoder-failure rate of $0.196$. These rates limit quantitative interpretation of the $d=9$ results.

  The sensitivity analysis [Fig.~\ref{fig:neural-guidance-sensitivity} and Table~\ref{tab:neural-guidance-summary}] varies the guidance strength at $d=5$. Increasing $\lambda$ from $0$ to $1$ reduces sampled mean LER from $0.1773$ to $0.1663$ and maximum LER from $0.4320$ to $0.4113$. No decoder failures or matching fallbacks occur in this sweep.

  These summaries should be read at two levels. The decoder ordering is a curve-level observation: Union-Find exhibits larger LER over the sampled moderate-to-high-noise hybrid region, and the targeted $d=9$ run preserves that ordering. The scalar crossing values require stricter qualification. The two matching-style pairwise estimates are separated because the available distances are small and the lower estimate lies in a low-LER region. The Union-Find $(d=5,d=7)$ result is weaker still because the dense grid contains no sign change. Reporting the localization method beside each value makes this difference visible instead of compressing distinct estimator outcomes into a single nominal threshold.

  \subsection{Comparison with Previous Studies}
  Table~\ref{tab:literature-comparison} places the LiDMaS+ results alongside representative published surface-code studies. The studies use different noise models, decoders, and endpoint metrics, so the values are not one-to-one benchmarks.

  \begin{table*}[t]
    \centering
    \small
    \caption{Contextual comparison between the present study and representative published surface-code studies.}
    \label{tab:literature-comparison}

    \begin{tabular}{llll}
      \toprule
      Study & Setting & Decoder & Key result \\
      \midrule

      Present study (LiDMaS+) &
      Simulated, $d=3,5,7$; targeted $d=9$ &
      Matching-style; UF; neural-guided &
      Pauli $p_c=0.0531$; transition-window hybrid crossings \\

      Wang (2011)\cite{wang2011surfaceover1} &
      Stochastic noise &
      Matching-style &
      $1.1$--$1.4\%$ threshold \\

      Fowler (2012)\cite{fowler2012proofmatching} &
      2D NN, noisy gates &
      Matching proof &
      $p<7.4\times10^{-4}$ \\

      Tuckett (2018)\cite{tuckett2018ultrahigh} &
      Biased noise &
      Tensor-network &
      $43.7\%$ (pure); $28.2\%$ at $\eta=10$ \\

      Tuckett (2020)\cite{tuckett2020faulttolerant} &
      Fault-tolerant bias &
      Bias-aware MWPM &
      $>6\%$ (dephasing); $\sim5\%$ at $\eta=100$ \\

      Google QAI (2023)\cite{acharya2023suppressing} &
      Superconducting, $d=3,5$ &
      Approx.\ maximum likelihood &
      $2.914\%$ (d=5) vs $3.028\%$ (d=3) \\

      Google QAI (2025)\cite{google2025belowthreshold} &
      Hardware to $d=7$ &
      Neural + real-time &
      $\Lambda=2.14$; $\epsilon_7=1.43\times10^{-3}$ \\

      Lee (2021)\cite{lee2021rectangular} &
      Rectangular vs square &
      MWPM &
      Failure ratio $4.37$ at $\eta=2.5$ \\

      \bottomrule
    \end{tabular}
  \end{table*}

  The range of values in Table~\ref{tab:literature-comparison} reflects differences in noise models, decoder assumptions, and reported metrics. The present study isolates decoder and estimator effects within one workflow.

  \section{Discussion}
  In the Pauli-reference mode, the matching-style backend yields lower LER, identifiable crossings, and a lower collapse cost than Union-Find. In the digitized hybrid mode, Union-Find degrades at larger distance and moderate-to-high $\sigma$. Neural-guided matching remains close to unguided matching at moderate noise but develops failure diagnostics at larger distance and high noise.

  The dense hybrid transition window replaces lower-boundary coarse-grid proxies with interior crossing estimates. The low-LER $(d=5,d=7)$ estimates remain sensitive to the sampled grid. At $d=9$, high matching-fallback rates limit quantitative interpretation.

  Several conclusions are robust within the implemented protocol. The Pauli-reference matching-style curves remain below the Union-Find curves, the matching-style Pauli-reference crossing and collapse summaries are mutually consistent at the reported resolution, and the hybrid Union-Find curves worsen relative to matching-style decoding as distance and $\sigma$ increase. Other conclusions are conditional. The hybrid pairwise crossings are not yet a single converged critical point, and the $d=9$ matching-style values mix exact-solver and greedy-fallback behavior. The neural-guidance sweep shows a modest improvement over the sampled window, but it does not establish superiority over production matching or analog-aware learned decoders.

  The separation between robust and conditional findings is useful for subsequent experiments. A production matching backend would remove the high-defect fallback as a confounding factor. Adaptive allocation of trials would increase resolution where pairwise curves are close without spending the same budget on plateaus. A broader distance set would test whether the low-LER hybrid crossings move toward a common value or remain pair-dependent. Finally, passing continuous likelihood information through the decoder interface would create a distinct analog-aware experiment rather than an incremental modification of the present digitized workflow.

  Threshold-oriented results should report the decoder, estimator, sweep resolution, confidence intervals, decoder-failure diagnostics, and backend fallback rate.

  \subsection{Scope and Validity Controls}
  All decoders run in the same LiDMaS+ harness with matched seeds, sweep grids, distance sets, and output/estimator paths. This design supports an internal comparison; it does not establish cross-platform performance.

  Finite sampling limits resolution in low-LER tails and near crossing regions. Fixed-distance sweeps use $3000$ trials per point, threshold-oriented sweeps use $2000$, the dense hybrid transition window uses $3000$, and the $d=9$ extension uses $1000$. Adaptive sampling near candidate crossings would reduce uncertainty.

  Distance coverage is limited to $d=3,5,7$, with a targeted extension at $d=9$. These distances do not support asymptotic or universality-class claims.

  The reported summaries depend on the LiDMaS+ noise-to-syndrome mapping, decoder interface, estimator definitions, and matching backend. The matching-style backend solves small instances exactly and applies a greedy fallback for larger defect sets; outputs include the fallback rate.

  The neural-guided model is lightweight and trained within this workflow. Its results are an internal sensitivity analysis, not a state-of-the-art learned-decoder benchmark.

  Within this scope, decoder choice changes threshold summaries, and hybrid scalar crossings remain estimator- and implementation-conditional.

  \section{Conclusion and Future Work}
  This study compares surface-code decoders under Pauli-reference and digitized hybrid continuous-variable/discrete noise within one LiDMaS+ workflow. The matching-style backend produces lower LER than Union-Find in the Pauli-reference mode and yields a stable threshold summary under the sampled protocol. In the hybrid mode, Union-Find degrades at fixed and larger distances. Neural-guided matching remains close to unguided matching at moderate noise but exhibits decoder-failure diagnostics at the largest sampled noise values.

  Threshold values are outputs of an inference pipeline, not decoder-free properties of a code and noise model alone. They depend on the decoding backend, sweep window, estimator, and statistical budget near the transition region. LiDMaS+ makes these choices comparable within one workflow.

  Future work should replace the high-defect fallback with a production matching backend, extend controlled scaling beyond $d=9$, allocate trials adaptively near candidate crossings, and validate learned models on independent datasets.

  \section*{Author Contributions}

  D.D.K.W: conceptualization, methodology, validation and visualization, software, writing – original draft, review \& editing. C.O: methodology, validation and visualization, software, writing – original draft, review \& editing.  L.G: methodology, validation and visualization, software, writing – original draft, review \& editing. S.G: methodology, validation and visualization, software, writing – original draft, review \& editing.

  \section*{Acknowledgment(s)}

  The authors acknowledge contributors, users, and reviewers who provided
  feedback on decoding workflows, reproducibility scripts, and
  documentation quality. Any opinions, findings, conclusions, or recommendations expressed in this research are those of the author(s) and do not necessarily reflect the views of their respective affiliations.

  \section*{Data \& Code Availability}
  The data generated and analyzed during the present study are included in this study and its supplementary materials. Supplementary code developed from the \texttt{LiDMaS+} simulator is available on \href{https://github.com/DennisWayo/lidmas_cpp}{GitHub} to support transparency and reproducibility.

  \section*{Funding}
  This research was not funded.

  \section*{Disclosure statement}

  No potential conflict of interest was reported by the author(s).

  \appendix
  \section{Derivations and Consistency Checks for the Experimental Pipeline}
  This appendix records the operational equations used in the study and short consistency proofs for the estimator definitions.

  \subsection{Noise Models and CV-to-Discrete Mapping}
  In the \texttt{paper\_01} Pauli-reference workflow, a single parameter $p$ controls Bernoulli X-fault sampling on data qubits, and the decoded syndrome branch is the corresponding $Z$-stabilizer parity channel:
  \begin{equation}
    e^{(X)}_i \sim \mathrm{Bernoulli}(p),\qquad
    s^{(Z)} = H_Z e^{(X)} \pmod 2,
    \label{eq:pauli-branch}
  \end{equation}
  Equation~\eqref{eq:pauli-branch} is followed by single-sector decoding and logical-failure estimation. This is a code-capacity-style internal reference mode, not a full circuit-level depolarizing schedule.

  In the hybrid mode, each qubit samples continuous displacement components
  \begin{equation}
    \delta q,\delta p \sim \mathcal{N}(0,\sigma^2),
    \label{eq:hybrid-gaussian}
  \end{equation}
  then digitizes using nearest-lattice rounding at scale $\sqrt{\pi}$:
  \begin{equation}
    n_q=\mathrm{round}\!\left(\frac{\delta q}{\sqrt{\pi}}\right),\quad
    n_p=\mathrm{round}\!\left(\frac{\delta p}{\sqrt{\pi}}\right),\quad
    e^{(X)}=\lvert n_q\rvert \bmod 2,\quad
    e^{(Z)}=\lvert n_p\rvert \bmod 2.
    \label{eq:hybrid-digitization}
  \end{equation}
  In this workflow, only the $e^{(X)}\!\to s^{(Z)}$ branch is decoded in the hybrid threshold loop, i.e.,
  \begin{equation}
    s^{(Z)} = H_Z e^{(X)} \pmod 2.
    \label{eq:hybrid-z-syndrome}
  \end{equation}
  Equations~\eqref{eq:hybrid-gaussian}, \eqref{eq:hybrid-digitization}, and \eqref{eq:hybrid-z-syndrome} define the digitized CV-to-discrete mapping used in this study. Therefore, no continuous soft-information vector is provided to the decoder backend in the presented hybrid runs.

  \subsection{Decoder Objectives}
  Let $\mathcal{X}$ be the observed defect set and let $w_{ij}$ be pair/boundary weights on the decoder graph.

  For small matching instances, the matching-style backend solves
  \begin{equation}
    M^\star=\arg\min_{M\in\mathcal{M}(\mathcal{X})}\sum_{(i,j)\in M} w_{ij},
    \label{eq:mwpm-objective}
  \end{equation}
  where $\mathcal{M}(\mathcal{X})$ is the set of admissible matchings, including boundary matches. Equation~\eqref{eq:mwpm-objective} defines the exact small-instance correction objective, and the correction chain is lifted from paths corresponding to $M^\star$. For larger defect sets, the present implementation applies a greedy fallback and records the fallback rate.

  Union-Find tracks cluster parity
  \begin{equation}
    \pi(C)=|C\cap\mathcal{X}|\;\mathrm{mod}\;2,
    \label{eq:uf-parity}
  \end{equation}
  as in Eq.~\eqref{eq:uf-parity}, then grows odd clusters, merges collisions, and peels spanning forests to produce a parity-consistent correction.

  Neural-guided MWPM uses a lightweight linear guidance model for weight scaling. For feature vector
    $x_{ij}=(d_{\mathrm{man}},\Delta x,\Delta y,b_{\partial})$, the scale factor is
  \begin{equation}
    \alpha_{ij}=\mathrm{clip}\!\left(
                                                 \beta_0+\beta_1 d_{\mathrm{man}}+\beta_2 \Delta x+\beta_3 \Delta y+\beta_4 b_{\partial},
                                                 \alpha_{\min},\alpha_{\max}\right),
    \label{eq:neural-alpha}
  \end{equation}
  and the guidance-strength interpolation is
  \begin{equation}
    \tilde{w}_{ij}(\lambda)=w_{ij}+\lambda\left(\alpha_{ij}w_{ij}-w_{ij}\right),\qquad \lambda\in[0,1].
    \label{eq:neural-weight}
  \end{equation}
  Equations~\eqref{eq:neural-alpha} and \eqref{eq:neural-weight} preserve the matching objective while changing edge priorities through learned linear reweighting. The sensitivity sweep evaluates $\lambda=0,0.5,1$.

  \subsection{Logical-Error and Threshold Estimators}
  For code distance $d$, sweep parameter $\theta$ (either $p$ or $\sigma$), and $T$ trials with $F$ logical failures,
  \begin{equation}
    \widehat{p}_{L}(d,\theta)=\frac{F}{T}.
    \label{eq:ler-estimate}
  \end{equation}
  The implementation reports Wilson-score confidence intervals. Starting from the estimator in Eq.~\eqref{eq:ler-estimate}, with $z=1.96$ and
    $\widehat{p}_{L}=F/T$, define
  \begin{equation}
    \mu=\frac{\widehat{p}_{L}+z^2/(2T)}{1+z^2/T},\qquad
    h=\frac{z}{1+z^2/T}\sqrt{\frac{\widehat{p}_{L}(1-\widehat{p}_{L})}{T}+\frac{z^2}{4T^2}}.
    \label{eq:wilson-params}
  \end{equation}
  The reported interval is $[\mu-h,\mu+h]$ from Eq.~\eqref{eq:wilson-params}.

  For two distances $(d_a,d_b)$ on grid points $\theta_k$, define
  \begin{equation}
    \Delta_k=\widehat{p}_L(d_a,\theta_k)-\widehat{p}_L(d_b,\theta_k).
    \label{eq:delta-curve}
  \end{equation}
  If $\Delta_k\Delta_{k+1}<0$ in Eq.~\eqref{eq:delta-curve}, a linearized crossing estimate is
  \begin{equation}
    \widehat{\theta}_c=
    \theta_k-\Delta_k\frac{\theta_{k+1}-\theta_k}{\Delta_{k+1}-\Delta_k}.
    \label{eq:crossing-linear}
  \end{equation}
  If no sign change exists on the sampled grid, the practical proxy is $\theta_{k^\star}$ with
  \begin{equation}
    k^\star=\arg\min_k |\Delta_k|.
    \label{eq:crossing-proxy}
  \end{equation}
  The transition-window hybrid analysis excludes each initial exact-zero low-noise plateau before applying Eqs.~\eqref{eq:crossing-linear} and \eqref{eq:crossing-proxy}. This avoids reporting a plateau endpoint as a physical crossing. Equation~\eqref{eq:crossing-linear} defines the linearized crossing estimator, and Eq.~\eqref{eq:crossing-proxy} defines the practical proxy used when an interior sign change is absent.

  For Pauli finite-size analysis, the standard collapse ansatz is
  \begin{equation}
    p_L(d,p)\approx f\!\left((p-p_c)d^{1/\nu}\right),
    \label{eq:collapse-ansatz}
  \end{equation}
  with $(p_c,\nu)$ obtained by minimizing a collapse cost over sampled curves and bootstrap resamples using Eq.~\eqref{eq:collapse-ansatz}.

  \subsection{Consistency Proof Sketches}
  \begin{claim}
    The interval defined by Eq.~\eqref{eq:wilson-params} is the Wilson score interval for a Bernoulli proportion with confidence parameter $z$.
  \end{claim}
  \begin{proof}
    Let $\hat p=\widehat p_L=F/T$ from Eq.~\eqref{eq:ler-estimate}. The score-test acceptance region for null proportion $p$ is
    \[
      \frac{(\hat p-p)^2}{p(1-p)/T}\le z^2.
    \]
    Rearrangement gives a quadratic inequality in $p$:
    \[
      (T+z^2)p^2-(2T\hat p+z^2)p+T\hat p^2\le 0.
    \]
    The solution set is the closed interval between the two roots. Writing the roots in center-radius form yields
    \[
      \mu=\frac{\hat p+z^2/(2T)}{1+z^2/T},\qquad
      h=\frac{z}{1+z^2/T}\sqrt{\frac{\hat p(1-\hat p)}{T}+\frac{z^2}{4T^2}},
    \]
    so the interval is $[\mu-h,\mu+h]$, exactly as reported in Eq.~\eqref{eq:wilson-params}.
  \end{proof}

  \begin{claim}
    Under linear interpolation between adjacent grid points $(\theta_k,\Delta_k)$ and $(\theta_{k+1},\Delta_{k+1})$, Eq.~\eqref{eq:crossing-linear} is the unique zero of the interpolant whenever $\Delta_k\Delta_{k+1}<0$.
  \end{claim}
  \begin{proof}
    Define the affine interpolant
    \[
      L(\theta)=\Delta_k+\frac{\Delta_{k+1}-\Delta_k}{\theta_{k+1}-\theta_k}(\theta-\theta_k).
    \]
    Solving $L(\theta)=0$ yields
    \[
      \theta=\theta_k-\Delta_k\frac{\theta_{k+1}-\theta_k}{\Delta_{k+1}-\Delta_k},
    \]
    which is exactly Eq.~\eqref{eq:crossing-linear}. Because $\Delta_k\Delta_{k+1}<0$, the endpoint values have opposite signs; therefore, the zero lies in $(\theta_k,\theta_{k+1})$. Uniqueness follows because $L$ is affine with nonzero slope when $\Delta_{k+1}\neq\Delta_k$.
  \end{proof}

  \begin{claim}
    If the seed map $g$ in Eq.~\eqref{eq:seed-map} is deterministic, repeated executions with identical $(s_0,d,\theta,t)$ generate identical pseudo-random streams.
  \end{claim}
  \begin{proof}
    For fixed inputs $(s_0,d,\theta,t)$, determinism implies $g$ returns the same seed value on every execution. A pseudo-random generator initialized with the same seed and the same call order produces the same sequence; thus, trial-level stochastic draws are reproducible under identical execution semantics.
  \end{proof}

  \subsection{Deterministic Reproducibility Parameterization}
  Each run is indexed by $(d,\theta,t)$ with base seed $s_0$. A deterministic seed map
  \begin{equation}
    s(d,\theta,t)=g(s_0,d,\theta,t)
    \label{eq:seed-map}
  \end{equation}
  guarantees identical pseudo-random streams for repeated executions under the same configuration. Together with fixed sweep sets
  \begin{equation}
    \begin{aligned}
      \mathcal{D} &= \{3,5,7\}, \\
      \mathcal{P} &= \{0.04,0.05,\dots,0.12\}, \\
      \Sigma &= \{0.05,0.10,\dots,0.60\}, \\
      \Sigma_{\mathrm{ref}} &= \{0.30,0.31,\dots,0.50\}, \\
      \mathcal{D}_{\mathrm{ext}} &= \{9\}.
    \end{aligned}
    \label{eq:sweep-sets}
  \end{equation}
  Equations~\eqref{eq:seed-map} and \eqref{eq:sweep-sets} define the reproducible threshold grids, including the dense transition-window sweep and targeted distance extension.

  \section{Supplementary Indexed Lattice and Decoder Trace}
  \label{app:indexed-trace}
  This appendix retains the indexed $d=5$ support graph and representative decoder trace as implementation-facing reproducibility aids. Figure~\ref{fig:indexed-syndrome-grid} shows the check/data indices used to interpret the request. Figure~\ref{fig:distance-syndrome-grid} shows the same support topology across the distances used in the study without full labels, since fully indexed $d=7$ and $d=9$ panels are not readable at journal scale. Figure~\ref{fig:decoder-trace-highlighted} marks the active Z-syndrome checks and correction qubits for the visualized $d=5$ instance, while Table~\ref{tab:decoder-trace-example} gives replayed single-instance trace lists for $d=3,5,7,9$.

  \begin{figure*}[t]
    \centering
    \includegraphics[width=\textwidth]{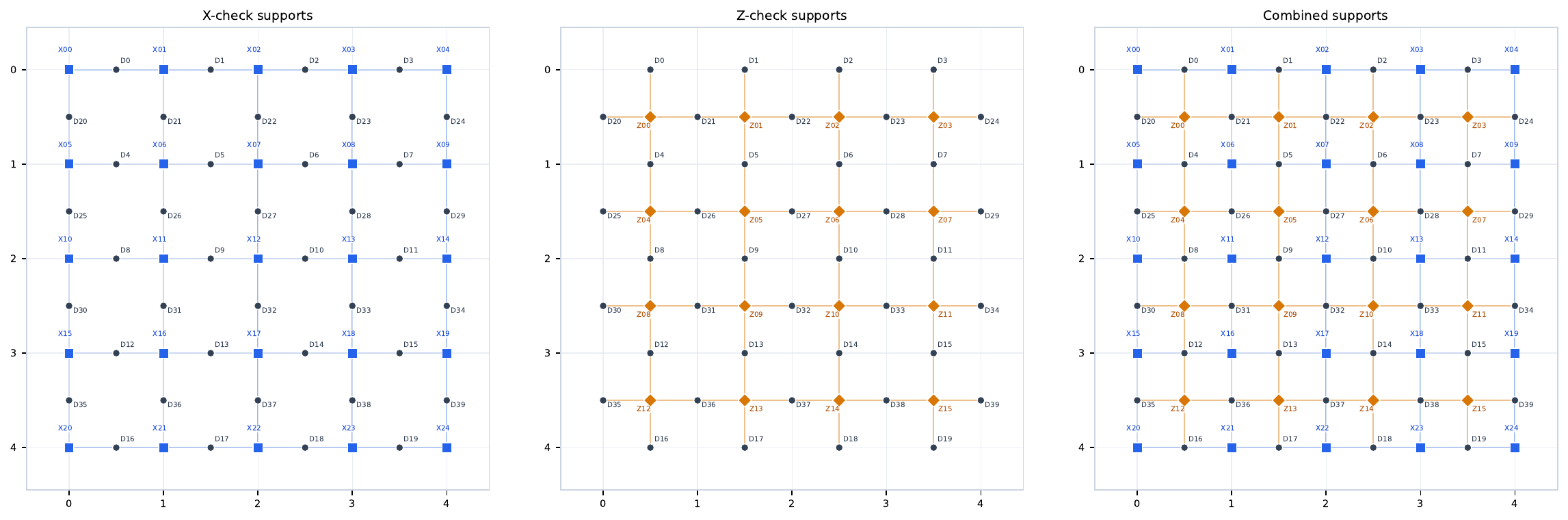}
    \caption{Indexed $d=5$ surface-code support graph used for decoder-trace interpretation. Panels show X-check supports, Z-check supports, and combined supports with explicit check/data indices. Data qubits are dark circles, X checks are blue squares, and Z checks are orange diamonds.}
    \label{fig:indexed-syndrome-grid}
  \end{figure*}

  \begin{figure*}[t]
    \centering
    \includegraphics[width=\textwidth]{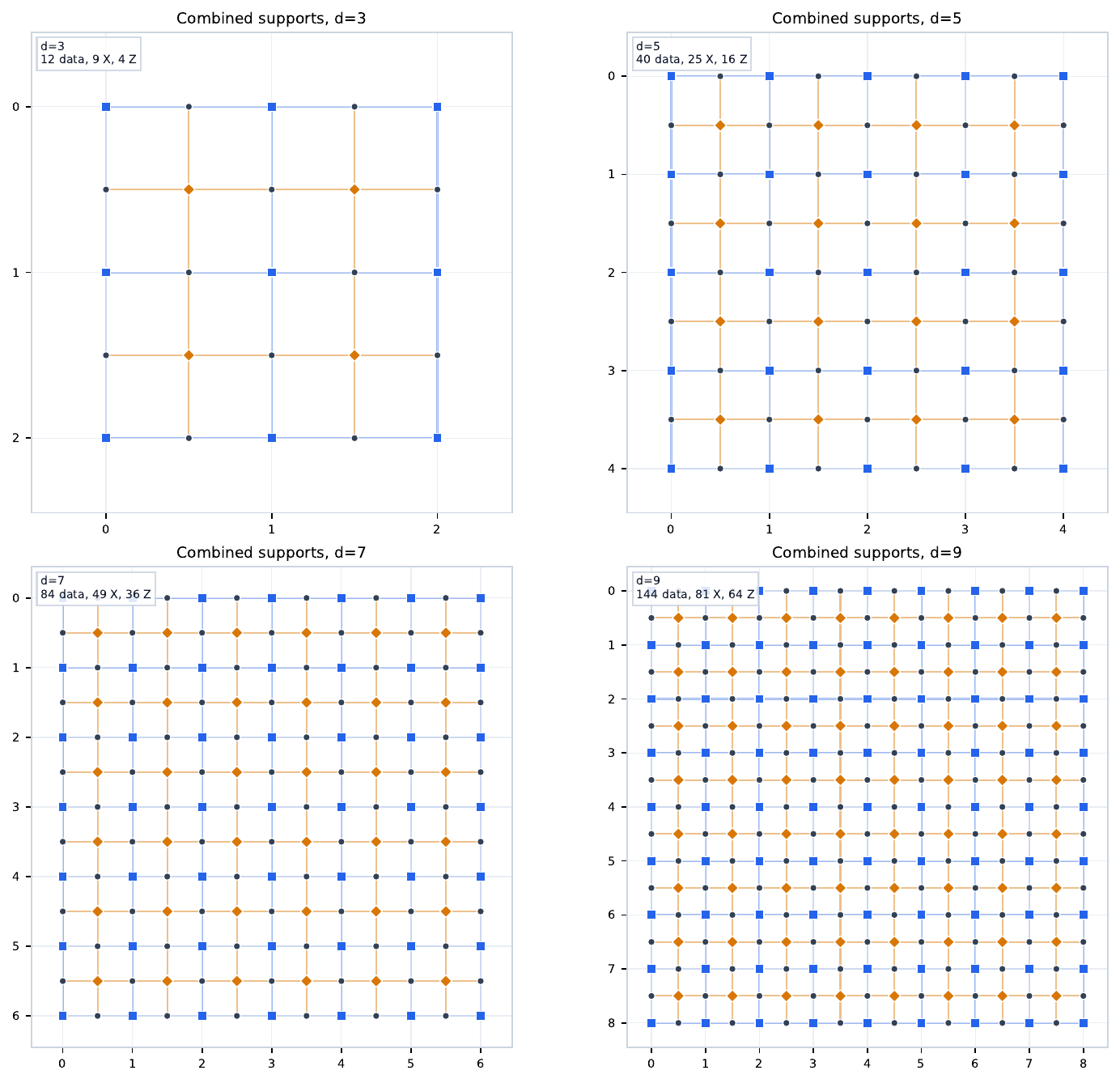}
    \caption{Combined support topology for $d=3,5,7,9$. The panels use the same data-qubit/check placement convention as Fig.~\ref{fig:indexed-syndrome-grid}; full labels are omitted for readability at larger distance.}
    \label{fig:distance-syndrome-grid}
  \end{figure*}

  \begin{figure*}[t]
    \centering
    \includegraphics[width=\textwidth]{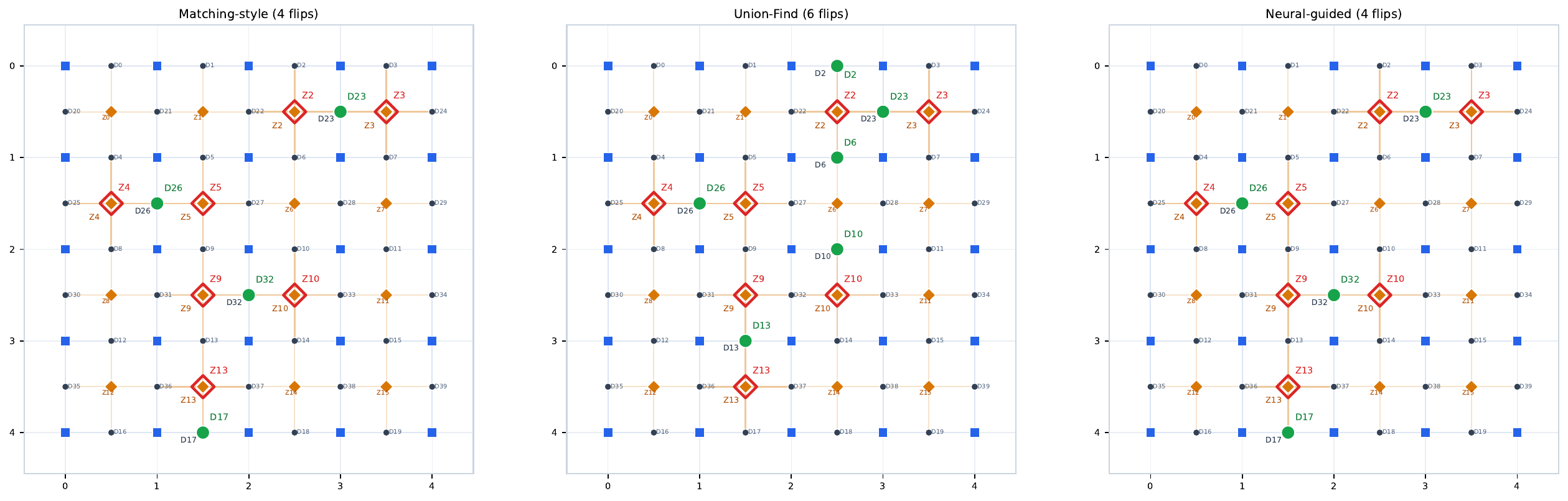}
    \caption{Highlighted single-instance decoder trace ($d=5$, hybrid $\sigma=0.18$). Red outlined diamonds mark active Z-syndrome checks; green data-qubit nodes mark decoder-selected correction flips.}
    \label{fig:decoder-trace-highlighted}
  \end{figure*}

  \begin{table*}[t]
    \centering
    \scriptsize
    \setlength{\tabcolsep}{3pt}
    \caption{Single-instance decoder traces across distances (hybrid-style decoder-IO requests, $\sigma=0.18$).}
    \label{tab:decoder-trace-example}
    \begin{tabular}{llll}
      \toprule
      Distance & Decoder & Detected Z-syndrome indices & Correction qubit indices \\
      \midrule
      $d=3$ & Matching-style & 0, 3 & 3, 7 \\
      $d=3$ & Union-Find & 0, 3 & 2, 10 \\
      $d=3$ & Neural-guided & 0, 3 & 6, 11 \\
      $d=5$ & Matching-style & \shortstack[l]{2, 3, 4, 5, 9,\\ 10, 13} & 17, 23, 26, 32 \\
      $d=5$ & Union-Find & \shortstack[l]{2, 3, 4, 5, 9,\\ 10, 13} & \shortstack[l]{2, 6, 10, 13, 23,\\ 26} \\
      $d=5$ & Neural-guided & \shortstack[l]{2, 3, 4, 5, 9,\\ 10, 13} & 17, 23, 26, 32 \\
      $d=7$ & Matching-style & \shortstack[l]{2, 5, 6, 8, 11,\\ 13, 15, 17, 18, 21,\\ 23, 24, 25, 26, 27,\\ 30, 32, 33, 34, 35} & \shortstack[l]{8, 21, 33, 38, 40,\\ 48, 49, 55, 56, 57,\\ 62, 63, 69, 70, 72,\\ 77, 83} \\
      $d=7$ & Union-Find & \shortstack[l]{2, 5, 6, 8, 11,\\ 13, 15, 17, 18, 21,\\ 23, 24, 25, 26, 27,\\ 30, 32, 33, 34, 35} & \shortstack[l]{8, 11, 12, 21, 23,\\ 24, 31, 57, 73, 78,\\ 80, 82} \\
      $d=7$ & Neural-guided & \shortstack[l]{2, 5, 6, 8, 11,\\ 13, 15, 17, 18, 21,\\ 23, 24, 25, 26, 27,\\ 30, 32, 33, 34, 35} & \shortstack[l]{8, 11, 13, 21, 23,\\ 24, 33, 50, 72, 78,\\ 79, 82} \\
      $d=9$ & Matching-style & \shortstack[l]{0, 2, 3, 4, 5,\\ 6, 9, 10, 13, 18,\\ 19, 23, 25, 30, 31,\\ 32, 33, 37, 38, 42,\\ 46, 47, 48, 49, 52,\\ 56, 59, 60, 61, 62} & \shortstack[l]{2, 3, 4, 6, 13,\\ 33, 38, 60, 67, 69,\\ 70, 72, 83, 93, 98,\\ 107, 108, 114, 115, 116,\\ 117, 118, 119, 124, 127,\\ 135} \\
      $d=9$ & Union-Find & \shortstack[l]{0, 2, 3, 4, 5,\\ 6, 9, 10, 13, 18,\\ 19, 23, 25, 30, 31,\\ 32, 33, 37, 38, 42,\\ 46, 47, 48, 49, 52,\\ 56, 59, 60, 61, 62} & \shortstack[l]{10, 14, 17, 26, 31,\\ 38, 44, 50, 52, 56,\\ 73, 74, 75, 77, 87,\\ 92, 93, 101, 109, 113,\\ 124, 128, 139, 141} \\
      $d=9$ & Neural-guided & \shortstack[l]{0, 2, 3, 4, 5,\\ 6, 9, 10, 13, 18,\\ 19, 23, 25, 30, 31,\\ 32, 33, 37, 38, 42,\\ 46, 47, 48, 49, 52,\\ 56, 59, 60, 61, 62} & \shortstack[l]{21, 29, 31, 33, 37,\\ 38, 42, 60, 73, 74,\\ 76, 78, 83, 93, 109,\\ 110, 124, 127, 136, 137,\\ 138, 141} \\
      \bottomrule
    \end{tabular}
  \end{table*}

  \clearpage
  \bibliographystyle{apsrev4-2}
  \bibliography{paper_01}

\end{document}